\documentclass[preprint,12pt]{elsarticle}

\usepackage{amssymb}
\usepackage{epsfig}
\usepackage{times}
\usepackage[T1]{fontenc}
\usepackage{makeidx}
\long\def\comment #1\commentend{} \long\def\commabs #1\commabsend{}
\long\def\commful #1\commfulend{#1}

\commabs
\commabsend

\usepackage{psfrag}

\newtheorem{theorem}{Theorem}
\newtheorem{lemma}{Lemma}
\newproof{proof}{Proof}
\newdefinition{definition}{Definition}

\newtheorem{remark1}{Remark}

\newcommand{\remove}[1]{}

\journal{Theoretical Computer Science}

\begin{document}

\begin{frontmatter}



\title{Faster Treasure Hunt and\\ Better Strongly Universal Exploration Sequences\tnoteref{tnot1}}
\tnotetext[tnot1]{A preliminary version of this paper appears in the Proceedings of the 18th International Symposium on Algorithms and Computation, ISAAC 2007.}

%\titlerunning{Faster Treasure Hunt, and Better Strongly Universal Exploration Sequences}

\author{Qin Xin\corref{cor1}\fnref{cor2}}
%
%\authorrunning{Q. Xin}

\address{Faculty of Science and Technology, University of the Faroe Islands, Torshavn, Faroe Islands}

\ead{qinx@setur.fo}

\cortext[cor1]{Correspondence to: Faculty of Science and Technology, University of the Faroe Islands, Noatun 3, FO 100, Torshavn, Faroe Islands. Tel.: +298 352575.}
\fntext[cor2]{Part of this research was performed while this author was a senior research fellow at Universit$\acute{e}$ catholique de
Louvain, Belgium, a scientist/postdoctoral research fellow at Simula Research Laboratory and a postdoctoral research fellow
at The University of Bergen, Norway. This work was partially supported by Visiting Scholarship of State Key Laboratory of
Power Transmission Equipment $\&$ System Security and New Technology (Chongqing University) (2007DA10512710408).}

%\address{}

\begin{abstract}
In this paper, we investigate the explicit deterministic treasure hunt problem
in a $n$-vertex network. This problem was firstly
introduced by Ta-Shma and Zwick in \cite{TZ07} [SODA'07].
Note also it is a variant of the well known rendezvous problem in which one of the robot
(the treasure) is always stationary.
In this paper, we propose an $O(n^{c(1+\frac{1}{\lambda})})$-time algorithm for the treasure hunt problem, which significantly
improves the currently best known result of running time $O(n^{2c})$ in \cite{TZ07},
where $c$ is a constant induced from the construction of an universal exploration sequence
in \cite{R05,TZ07}, and $\lambda \gg 1$ is an arbitrary large, but fixed,  integer constant. 
The treasure hunt problem also motivates the
study of strongly universal exploration sequences. In this paper, we also propose a much better explicit construction for strongly universal exploration sequences compared to the one in \cite{TZ07}.
\end{abstract}


\end{frontmatter}

% \linenumbers

\section{Introduction}

In the {\sl rendezvous problem} (\cite{DFKP06,KM06,TZ07}),
two robots are placed in an unknown environment
modeled by a finite, connected, undirected graph $G =(V,E).$
We assume that $|V | = n.$ The size of the
network, i.e., the number of vertices in the graph
is not known to the robots.
The edges incident on a vertex $u\in V$ are numbered
$0, 1, 2,\cdots, deg(u)-1,$ in
a predetermined manner, where $deg(u)$ is the degree
of $u.$ In general, the numbering is not assumed to be
consistent, i.e., an edge $(u, v)\in E$ may be the $i$-th edge
of $u$ but the $j$-th edge of $v$, where $i\not= j.$

When a robot is in a vertex $u\in V$ it is told the
degree $deg(u)$ of $u.$ However, all vertices of the same degree are
not distinguishable. The robots are not allowed to
put any information such as tokens or markers at the vertices that they visit.
At any time step a robot is only allowed to either traverse an
edge, or stay in place. When the robot is at a vertex $u$ it
may ask to traverse the $i$-th edge $(u, v) \in E$ of $u$, where
$0 \le i \le deg(u)-1.$ The robot observes itself at vertex $v,$ the
another endpoint of this edge. As described before, the $i$-th edge
of $u$ is the $j$-th edge of $v,$ for some $0 \le j \le deg(v)-1.$ In
general $j \not= i.$

There are two different variants of the model used in the field.
In the first one,
the robot is told the index $j$ of the edge it used to
enter $v$.
This allows the robot to return to $u$ at the next
step, if it wants to do so. This variant of the problem
is called the rendezvous problem with backtracking. In the second
variant of the model, the robot observes
itself at vertex $v$ without knowing which edge it used to get
there. We call this variant of the problem the general
rendezvous problem.

Same as most of work \cite{DFKP06,KM06,TZ07} described,
the main strategy is to give the two robots deterministic
sequences of instructions which will gurrantee that
two robots
would eventually meet each other, no matter in which graph they
are located, and no matter when they are activated. It is,
however, expected that such a meeting would happen
as soon as possible.
A robot is unaware of the whereabouts of another
robot, even it is very close to another one
in the graph. The two robots  meet only
when they are both active and are at same vertex
at same time. In particular, the two robots may
traverse the same edge but in different directions, and still
miss each other.

For the deterministic solutions, it has to be assumed
that two robots have different labels,
e.g. $L_1 \not= L_2.$ Without such an assumption
there is no deterministic way of breaking symmetry and
no deterministic strategy is possible.
An example has been shown in \cite{TZ07} if the two robots are completely
identical. Assume $G$ is a ring on $n$ vertices and
that the edges are labeled so that out of every vertex,
edge $0$ goes clockwise, while edge $1$ goes anti-clockwise.
If the two robots start the same time at different
vertices and follow the same instructions, they would
never meet! Same as the previous work \cite{TZ07}, We also assume that
the moves of the two robots are synchronous after both of them are activated.
The crutial feature of this problem is
that the two robots may be activated at different times which decides
arbitrarily by the adversary. A meeting can happen only when both robots
are active. The time complexity of any solution is bounded by
the number of steps used to complete such a task, which counts from the
activation of the second robot.

The {\sl treasure hunt problem} is a variant of the rendezvous
problem in which the robots are assigned the labels
0 and 1 and robot 0, the treasure, cannot move, which firstly introduced in \cite{TZ07}.
As in the rendezvous problem, the treasure and the seeking robot
 are not necessarily activated at the same
time.

\subsection{Previous work}

Dessmark {\sl et al.} \cite{DFKP06}
presented a deterministic solution of the rendezvous
problem which guarantees a meeting of the two
robots after a number of steps which is polynomial in $n$,
the size of the graph, $l$, the length of the shorter of the
two labels, and $\tau$, the difference between their activation
times.  More specifically, the bound on
the number of steps that they obtain is $\tilde{O}(n^5\sqrt{\tau l}+n^{10}l).$
In the same paper, Dessmark {\sl et al.} \cite{DFKP06} also
ask whether it is possible to obtain
a polynomial bound that is independent of $\tau$.
Kowalski and Malinowski \cite{KM06} have recently presented
a deterministic solution to the rendezvous
problem that guarantees a meeting after at
most $\tilde{O}(n^{15}+l^3)$ steps, which is independent of $\tau$, and also
firstly answer the open problem of \cite{DFKP06} when backtracking is allowed.
Very recently in \cite{TZ07}, Ta-Shma, and Zwick
propose a deterministic solution that guarantees a rendezvous
within $\tilde{O}(n^5l)$ time units after the activation of the second
robot, and also uses backtracking. This is the currently best known solution.
All the solutions mentioned above rely on the existence of a
universal traversal sequences, introduced by Aleliunas
{\sl et al.} \cite{AKL+79}, and are therefore non-explicit.
The first explicit  solution for both rendezvous problem and treasure hunt
problem can be found in \cite{TZ07}.
This work allows backtracking, by using the explicit
construction of a {\sl strongly universal exploration sequence SUES}. The time complexity
of the solutions for both problems is $O(n^{2c})$,  where
$c$ is a huge constant.
Other variants of the rendezvous problem could be found in \cite{DGK+06}.

Note if
randomization is allowed, then both the rendezvous problem and the treasure hunt problem
have the trivial solutions using a polynomial number of steps in term of the size of the
graph with high probability, e.g. a random walk by Coppersmith {\sl et al.} \cite{CTW93}.

\subsection{Our results}

We mainly study here the explicit deterministic treasure hunt problem with backtracking
in a $n$-vertex network.
It is a variant of the well known rendezvous problem in which one of the robot
(the treasure) is always stationary.
We propose an $O(n^{c(1+\frac{1}{\lambda})})$-time algorithm for treasure hunt problem, which significantly
improves currently best known result with running time $O(n^{2c})$ in \cite{TZ07},
where $c$ is a constant induced from the construction of an universal exploration sequence
in \cite{R05,TZ07}, $\lambda \gg 1$ is a constant integer.
The treasure hunt problem also motivates the
study of strongly universal exploration sequences.
In this paper, we also propose a much better explicit construction
for strongly universal exploration sequences (SUESs) compared to the one in \cite{TZ07}.
The improved explicit {\sl SUESs} could be also used to improve the explicit solution of the rendezvous problem in \cite{TZ07}. 

\section{The treasure hunt problem}

The treasure hunt problem is a variant of the well known
rendezvous problem in which one of the robots, the treasure,
is always stationary.
A seeking robot and the treasure are  placed in an unknown
location in an unknown environment, modeled again
by a finite, connected, undirected graph.  Same as
in the rendezvous problem, the treasure and the seeking robot
searching for it are not necessarily activated at the same
time. The most difficult case of the problem is when the seeking robot is
activated before the treasure.

To clarify our presentation, we give a formal definition for the
treasure hunt problem, which is a modified version from the rendezvous
problem by Ta-Shma and Zwick \cite{TZ07}.

Formally, a deterministic solution for the general
treasure hunt problem (without backtracking) is a deterministic
algorithm that computes a function $f : Z^+\times Z^+ \rightarrow Z^+,$
where for $d\ge 1$ and $t\ge 0$ we have $0\le f(d, t)\le d-1.$
This function defines the walk carried out by the seeking robot
as follows: at the $t$-th time unit since
activation, when at a vertex of degree $d,$
use edge number $f(d, t)$ to walk in the next step.

A deterministic solution for the treasure hunt problem
with backtracking is a deterministic algorithm that
computes a function $f : Z^+\times Z^+ \times (Z^+)^\ast\rightarrow Z^+,$
where for every $d\ge 1$, $t\ge 0$, and $T \in (Z^+)^\ast$ we
have $0\le f(d, t, T)\le d-1.$ This function defines the
walk carried out by the seeking robot as follows:
at the $t$-th time unit since activation, when at a vertex
of degree $d,$ if the sequence of edge numbers assigned
to the edges that were used to enter the vertices at the
previous time units is $T \in (Z^+)^{\ast}$ the robot will
exit the current node
using the edge number $f(d, t, T),$ in the next step.
In our solutions that use this model, the
function $f$ depends on $T$ only through its last element, which is the same as
\cite{TZ07}.

Throughout most of this paper we shall assume that
the graph $G$ in which the robots (seeking robot and the treasure robot)
are placed is a $d$-regular
graph, for some $d\ge 3$. Note it is easy to
extend the solutions given for the $d$-regular graphs to
general graphs using the ideas from \cite{DFKP06}.

\section{Universal and strongly universal exploration sequences}

For clarity of presentation, we use the same definitions as in \cite{TZ07}.

Let $G = (V,E)$ be a $d$-regular graph.
A sequence $\tau_1\tau_2\cdots\tau_k \in \{0,1, 2,$
$\cdots,d-1\}^k$ and
a starting  edge $e_0 = (v_{-1}, v_0)\in E$
define a walk $v_{-1}, v_0, \cdots , v_k$ as follows: For $1\le i\le k,$
if $(v_{i-1}, v_i)$ is the $s$-th edge of $v_i$, let $e_i = (v_i, v_{i+1})$ be
the $(s+\tau_i)$-th edge of $v_i$, where we assume here that the
edges of $v_i$ are numbered $0, 1,\cdots , d-1,$ and that $s+\tau_i$
is computed by modulo $d.$

\begin{definition}
(Universal Exploration Sequences
(UESs) \cite{K02,TZ07})\\ A sequence $\tau_1\tau_2\cdots\tau_l\in
\{0, 1, \cdots, d - 1\}^l$ is a universal exploration sequence
for $d$-regular graphs of size at most $n$ if for every
connected $d$-regular graph $G = (V,E)$ on at most $n$
vertices, any numbering of its edges, and any starting
edge $(v_{-1}, v_0) \in E,$ the walk obtained visits all the
vertices of the graph.
\end{definition}

Reingold \cite{R05} obtains
an explicit construction of polynomial-size {\sl UES:}

\begin{theorem}\label{R05}

(\cite{R05}) There exists a constant
$c\ge 1$ such that for every $d\ge 3$ and $n\ge 1$, a UES of
length $O(n^c)$ for $d$-regular graphs of size at most $n$ can
be constructed, deterministically, in polynomial time.

\end{theorem}

\begin{definition}

(Strongly Universal Exploration
Sequences (SUESs) \cite{TZ07}) A possibly infinite
sequence $\tau=\tau_1\tau_2\cdots,$ where $\tau_i \in \{0, 1,\cdots, d - 1\},$ is
a strongly universal exploration sequence (SUES) for
$d$-regular graphs with cover time $p(\cdot),$ if for any $n\ge 1,$
any contiguous subsequence of $\tau$ of length $p(n)$ is a
UES for $d$-regular graphs of size $n.$

\end{definition}

Let $O(n^c)$ be the length of a UES (see \cite{R05,TZ07}),
the main Theorem of this section shows that strongly
universal exploration sequences (SUESs) do exist and
they can be constructed deterministically in polynomial
time with cover time $p(n)=O(n^{c(1+\frac{1}{\lambda})})$, which significantly improves
the currently best known result in \cite{TZ07} with
$p(n)=O(n^{2c})$, where $c$ is a constant induced from the construction of an universal
exploration sequence in \cite{R05,TZ07}, and $\lambda \gg 1$ is an arbitrary large, but fixed,  integer constant.

In this section, we firstly give a weak solution with $p(n)=O(n^{{\frac{3}{2}}c})$, then we show our main
result described above.

\subsection{Explicit SUESs with $p(n)=O(n^{{\frac{3}{2}}c})$}

In this section, we propose a new explicit strongly
universal exploration sequence with cover time $p(n)=O(n^{{\frac{3}{2}}c})$ for
the $d$-regular graphs of size  at most $n,$ where $c$ is the same constant as was used in \cite{R05,TZ07}. It is a weak version (a special case) of our main result in Section \ref{mainresult}, but which gives more intuition on our new approaches in Section
\ref{mainresult}.

\subsubsection{Properties of exploration sequences}

A very useful property of exploration sequences \cite{K02,TZ07} is
that walks defined by an exploration
sequence can be reversed. For $\tau =\tau_1\tau_2\cdots\tau_k$ $\in
\{0, 1,\cdots, d - 1\}^k,$ we let $\tau^{-1} = \tau_k^{-1}\tau_{k-1}^{-1}\cdots
\tau_1^{-1},$ where $\tau_i^{-1}= d - \tau_i$ (mod $d$). It is not difficult to see that a walk
defined by an exploration sequence $\tau$ can be backtracked
by executing the sequence $0\tau^{-1}0.$ Note that
if $e_0, e_1,\cdots, e_k$ is the sequence of edges defined by $\tau$,
starting with $e_0,$ then executing $0\tau^{-1}0,$ starting with $e_k$
defines the sequence $e_k, \tilde{e}_k,\tilde{e}_{k-1},\cdots,\tilde{e}_0,e_0,$ where
$\tilde{e}$ is
the reverse of edge $e.$
Also, if $\tau$ is a universal exploration
sequence for graphs with size at most $n,$ then so
is $0\tau^{-1}$ starting with the last edge defined by $\tau$.

\subsubsection{Construction of SUESs}\label{wm}

Let $U_n$ be a sequence of length $n$ which is a universal
exploration sequence for $d$-regular graphs of size at most
$bn^{\frac{1}{c}},$ for some constants $b>1,$ and $c>1,$
which can be constructed, deterministically, in polynomial time of $n$
due to Theorem \ref{R05}.
We are interested in sequences $U_n$ only if $n$ is a power of $2.$
For the sake of technique, we can construct $U_n = U_1U_1U_2U_4\cdots U_{\frac{n}{2}}.$ Therefore, $U_k$ is a prefix of $U_n,$ for every $k = 2^i$ and $n = 2^j$, where $i < j.$

A strongly
universal exploration sequence $S_n$ is a sequence defined in a recursive manner.
Our approach is based on the similar idea in
\cite{TZ07}, but different interleaving components between the
symbols which originate from $U_n$.  We begin with
$S_1 = U_1.$ Assume that $U_n = u_1u_2\cdots u_n$ and that $n\ge 2.$

Define,
$$ S_n =u_1 S_{r_1} 0 S_{r_1}^{-1} 0 u_2 S_{r_2} 0 S_{r_2}^{-1} 0 u_3\cdots
u_i S_{r_i} 0 S_{r_i}^{-1} 0 u_{i+1}\cdots u_{n-1} S_{r_{n-1}} 0 S_{r_{n-1}}^{-1}0 u_n,$$

\noindent
for every $1\le i <n,$ we set $r_i =\langle i  \rangle$, where $\langle i\rangle
= \max\{2^{j}|2^{{\frac{3}{2}}j} \le i, j\in Z^{+}\}$. 
Similar construction strategy for $U_n$ (e.g., $U_n = U_1U_1U_2U_4\cdots U_{\frac{n}{2}}$) is also adopted to construct $S_n$.
Note that the sequence $S_k$ is a prefix of $S_n$ for every $k=2^p$ and $n=2^j$, where $p<j.$ Moreover, we also assign $r_i=r_{n-i}$, for every $1\le i <n$.
  
Furthermore, the sequence $S^{-1}_n$ differs with $S_n$ only on the symbols
that originate from $U_n$ and in the alignment of the $0'$s:
$$ S^{-1}_n =u_n^{-1} 0 S_{r_1} 0 S_{r_1}^{-1}  u^{-1}_{n-1}\cdots
u^{-1}_{i+1} 0S_{r_{n-i}} 0 S_{r_{n-i}}^{-1} u^{-1}_{i}\cdots u^{-1}_{2} 0 S_{r_{n-1}} 0 S_{r_{n-1}}^{-1} u^{-1}_1.$$

Note that $r_i\le \sqrt[3]{i^2}$.  Thus, if $r_{\frac{n}{2}}=\sqrt[3]{({\frac{n}{2}})^2},$   the first half of $S_n$ is equal to
$S_{\frac{n}{2}}S_{\sqrt[3]{({\frac{n}{2}})^2}}0,$ and ends with a full copy
of $S_{\sqrt[3]{({\frac{n}{2}})^2}}$, followed by a $0$. Similarly, the second
half of $S_n$ starts with a full copy of $S_{\sqrt[3]{({\frac{n}{2}})^2}}^{-1}.$
In the following, we bound the length of $S_n$.

\begin{lemma}

For every $n = 2^j,$ where $j\ge 1,$ $|S_n| < 258n.$

\end{lemma}

\begin{proof}
Let $s_n = |S_n|.$ It is not difficult to see that
$s_1 = 1,~s_2 = 6,~ s_4 = 16,~ s_8 = 46,~s_{16} = 126,~s_{32}=286,\cdots,s_{256}=3426.$
The claim that $s_n\le 258n$ for every
$n\ge 512$ then follows by using simple induction.
It is not difficult to see that
\begin{eqnarray*}
|S_{2^i}| &=& |U_{2^i}|+ (|U_{2^i}|-1)\cdot 2(|S_1|+1)+
\sum_{j=1}^{\lfloor\frac{2i}{3}\rfloor-1} \lfloor\frac{|U_{2^i}|-1}{2^{\frac{3j}{2}}}\rfloor
\cdot 2(|S_{2^j}|-|S_{2^{j-1}}|) \\
s_{2^i} &\le& 2^i+ 4\cdot 2^i + 2^i\cdot
\sum_{j=1}^{\lfloor\frac{2i}{3}\rfloor-1} \frac{2(s_{2^j}-s_{2^{j-1}})}{2^{\frac{3j}{2}}}
 \\
&=&
2^i+ 4\cdot 2^i + 2^i\cdot
(\sum_{j=1}^{4} \frac{2(s_{2^j}-s_{2^{j-1}})}{2^{\frac{3j}{2}}}
+\sum_{j=5}^{\lfloor\frac{2i}{3}\rfloor-1} \frac{2(s_{2^j}-s_{2^{j-1}})}{2^{\frac{3j}{2}}})
\\
&\le& 5\cdot 2^i + 2^i\cdot
\sum_{j=1}^{4} \frac{2(s_{2^j}-s_{2^{j-1}})}{2^{\frac{3j}{2}}}+ 2^i\cdot
\sum_{j=5}^{\lfloor\frac{2i}{3}\rfloor-1}
\frac{2(s_{2^{j-1}}+2s_{2^{\frac{2(j-1)}{3}}}+2)}{2^{\frac{3j}{2}}} \\
&\le& 5\cdot 2^i +
2^i\cdot
\sum_{j=1}^{4} \frac{2(s_{2^j}-s_{2^{j-1}})}{2^{\frac{3j}{2}}}+
2^i\cdot
\sum_{j=5}^{+\infty}
\frac{2(s_{2^{j-1}}+2s_{ 2^{\frac{2(j-1)}{3}}}+2)}{2^{\frac{3j}{2}}}\\
&\le& 5\cdot 2^i +
2^i\cdot
\sum_{j=1}^{4} \frac{2(s_{2^j}-s_{2^{j-1}})}{2^{\frac{3j}{2}}}+
2^i\cdot
\sum_{j=5}^{+\infty}
\frac{2(s_{2^{j-1}}+2s_{2^{\frac{2(j-1)}{3}}})}{2^{\frac{3j}{2}}}+
2^i\cdot
\sum_{j=5}^{+\infty}
\frac{4}{2^{\frac{3j}{2}}}
\\
&<& 6\cdot 2^i +
2^i\cdot
\sum_{j=1}^{4} \frac{2(s_{2^j}-s_{2^{j-1}})}{2^{\frac{3j}{2}}}+
2^i\cdot
\sum_{j=5}^{+\infty}
\frac{2(s_{2^{j-1}}+2s_{ 2^{\frac{2(j-1)}{3}}})}{2^{\frac{3j}{2}}}
\\
&\le& 6\cdot 2^i +
2^i\cdot (5+\frac{35\sqrt{2}}{8})+
2^i\cdot
\sum_{j=5}^{+\infty}
\frac{2(s_{2^{j-1}}+2s_{ 2^{\frac{2(j-1)}{3}}})}{2^{\frac{3j}{2}}}
\\
&\le& 6\cdot 2^i +
12\cdot 2^i+
2^i\cdot
\sum_{j=5}^{+\infty}
\frac{2(s_{2^{j-1}}+2s_{ 2^{\frac{2(j-1)}{3}}})}{2^{\frac{3j}{2}}}
\\
&\le& 18\cdot 2^i +
2^i\cdot
\sum_{j=5}^{+\infty}
\frac{2(258\cdot{2^{j-1}}+2\cdot 258\cdot 2^{\frac{2(j-1)}{3}})}{2^{\frac{3j}{2}}}
\\
&\le& 18\cdot 2^i +
2^i\cdot 258\cdot(\frac{2^{-\frac{5}{2}}}{1-\frac{1}{\sqrt{2}}}+\frac{2^{-{\frac{17}{6}}}}
{1-\frac{1}{2^{\frac{5}{6}}}}) \\
&<& 18\cdot 2^i + 2^i\cdot 258\cdot(0.93) \\
&=& (257.94)\cdot 2^i\\
&<& 258\cdot 2^i.
\end{eqnarray*}

\end{proof}

The sequence $S_n$ possesses the following interesting
combinatorial property:

\begin{lemma}\label{3:2}

Let $k$ and $n\ge 2 k^{\frac{3}{2}}$ be powers of $2.$ Then,
every subsequence $T$ of $S_n$ or $S_n^{-1}$
of length $s_{2k^{\frac{3}{2}}}+1 =
|S_{2k^{\frac{3}{2}}}|+1\le 516k^{\frac{3}{2}}$ contains, as a contiguous subsequence,
$S_k$ or  $0 S^{-1}_k.$

\end{lemma}

\begin{proof}

We prove the claim by induction on $n.$ If
$n = 2k^{\frac{3}{2}}$ then the claim is vacuously satisfied as $S_n$ contains
a full $S_k$.

Assume, therefore, that the claim holds for every
$m = 2^{j'}$
that satisfies $2k^{\frac{3}{2}}\le m < n = 2^j$. We show
that it also holds for $n.$ Let $T$ be a subsequence of $S_n$
of length $s_{2k^{\frac{3}{2}}} +1.$ Essentially the same argument works
if $T$ is a subsequence of such length of $S^{-1}_n$.
We use the {\bf exactly same} arguments as in Lemma $6.2$ \cite{TZ07}.
For completeness of our presentaion, we reproduce the analysis from \cite{TZ07}.

We consider the following cases:

\noindent
{\bf Case 1:} T is completely contained in a subsequence $S_m$
or $S^{-1}_m$ of $S_n,$ for some $m < n.$

The claim then follows immediately from the induction
hypothesis.

\noindent
{\bf Case 2:} $T$ is completely contained in a subsequence
$S_m0S^{-1}_m$ of $S_n,$ for some $m < n.$

In this case, $T = T'0T''$, where $T'$ is a suffix of $S_m$
and $T''$ is a prefix of $S^{-1}_m.$ Either
$|T'|\ge {\frac{1}{2}}s_{2k^{\frac{3}{2}}}$ or
$|T''|\ge {\frac{1}{2}}s_{2k^{\frac{3}{2}}}$.
Assume that
$|T''|\ge {\frac{1}{2}}s_{2k^{\frac{3}{2}}}$.
Another case is analogous. As $T''$ is a prefix of $S^{-1}_m$,
and  $|T''|\ge {\frac{1}{2}}s_{2k^{\frac{3}{2}}}$, it follows that
$m\ge 2k^{\frac{3}{2}}$. Now, $S^{-1}_k$
is almost a prefix of $S^{-1}_m$, in the sense that they differ
only in symbols that originate directly from $S_m.$
In particular, a prefix of $S^{-1}_m$ of length
${\frac{1}{2}}s_{2k^{\frac{3}{2}}}$, half the
length of $S_{2k^{\frac{3}{2}}}$, ends with a full copy of $S_k,$
followed by $0.$

\noindent
{\bf Case 3:} $T$ contains a symbol $u_l$ of $S_n$ that originates
from $U_n.$

In this case, $T = T'u_lT''.$ Again, we have either
$|T'|\ge {\frac{1}{2}}s_{2k^{\frac{3}{2}}}$ or
$|T''|\ge {\frac{1}{2}}s_{2k^{\frac{3}{2}}}$. Assume again
that $|T''|\ge {\frac{1}{2}}s_{2k^{\frac{3}{2}}}$.
Another case is analogous. Let
$$S_{n,l}= u_l S_{r_l}0S^{-1}_{r_l}0 u_{l+1}\cdots u_{n-1} S_{r_{n-1}}
0S^{-1}_{r_{n-1}}0 u_n$$
\noindent
be the suffix of $S_n$ that starts with the symbol $u_l$ that
originates from the $l$-th symbol of $U_n.$ We claim that
the prefix of $S_{n,l}$ of length ${\frac{1}{2}}s_{2k^{\frac{3}{2}}}$
contains a copy of $S_k.$
Let $l'=\lceil\frac{l}{k^{\frac{3}{2}}}\rceil k^{\frac{3}{2}}$ be
the first index after $l$ which is
divisible by $k^{\frac{3}{2}}.$ Clearly $r_{l'}\ge k$ and hence $S_k$ is a
prefix of $S_{r_{l'}}$. Thus,
$S' = u_l S_{r_l}0S^{-1}_{r_l}0\cdots u_{l'}S_k$ is a
prefix of $S_{m,l}$ which ends with a complete $S_k.$ As
for every $l\le i <l'$ we have $r_i = r_{i~mod~k^{\frac{3}{2}}},$ we have
that $S'$ is contained in the first half of $S_{2k^{\frac{3}{2}}},$ and hence
$|S'|\le {\frac{1}{2}}s_{2k^{\frac{3}{2}}}$ as expected.

\end{proof}

We are now ready to prove the following Theorem.

\begin{theorem}\label{wt}

If for any $n\ge 1$ of the power of $2$ there exists an UES of
length $O(n^c)$ for a $d$-regular graph of size at most $n$,
then there is an infinite SUES for this $d$-regular graph with
cover time $p(n)=O(n^{{\frac{3}{2}}c})$, where $c$ is the fixed constant from
the construction of an universal exploration sequence in \cite{R05,TZ07}.
Furthermore, the SUESs can be constructed deterministically in polynomial time.

\end{theorem}

\begin{proof}

Let us look at the recursive definition of $S_n$ and
ignore all the recursive components of $S_j$ such as $j < n$, and
their inverses, which because that
$0S^{-1}_j0$
reverses the actions of $S_j.$ The left parts are
$U_n=u_1,u_2,u_3,\cdots, u_n.$ However, note that $U_n$ is a UES for
for $d$-regular graphs of size at
least $bn^{\frac{1}{c}},$ for some constants $b\ge 1$ and $c>1$
due to Theorem~\ref{R05}. Theorem~\ref{R05} also show that such a
UES can be constructed, deterministically, in polynomial time.
According to Lemma~\ref{3:2}, we know that
every subsequence $T$ of the SUES we constructed of length
$s_{2n^{\frac{3}{2}}}+1 = O(n^{\frac{3}{2}})$ contains,
as a contiguous subsequence,
a full copy of $S_n$. Consequently,
there is an infinite SUES for $d$-regular graphs with
cover time $p(n)=O(n^{{\frac{3}{2}}c})$, where $c$ is the fixed constant from
the construction of an universal exploration sequence in \cite{R05,TZ07}.
Furthermore, the SUESs can be constructed deterministically in polynomial time.

\end{proof}

This thus gives us an explicit
solution to the treasure hunt problem. In fact, the seeking robot
just need run the SUES. The adversary will decide when
the treasure is put into the graph. But note that the subsequence of a SUES with
length $p(n)$ starting at the activation point forms a UES, and then the seeking robot finds the treasure by following the instruction of the sequence.

\subsection{Explicit SUESs with $p(n)=O(n^{c(1+\frac{1}{\lambda})})$}
\label{mainresult}

In this section, we propose our main result, a new explicit strongly
universal exploration sequence with cover time
$p(n)=O(n^{c(1+\frac{1}{\lambda})})$, which significantly improves the currently best
known result in \cite{TZ07} with
$p(n)=O(n^{2c})$, where $c$ is a fixed constant induced from
the construction of an universal exploration sequence in \cite{R05,TZ07},
and $\lambda \gg 1$ is an arbitrary large, but fixed,  integer constant.

\subsubsection{Treasure sequences}
A treasure sequence is an infinite sequence $Q(\lambda)= q_1, q_2, q_3, q_4,\cdots,$
based on a constant integer
$\lambda \gg 1$ such as
$$q_i= \sum_{j=i}^{+\infty} (2^{\frac{-j}{\lambda}} + 2^{-(\frac{2\lambda+1}{\lambda^2+\lambda}\cdot j+ \frac{\lambda}{\lambda+1}-2)} +
2^{-(\frac{\lambda+1}{\lambda}\cdot j -2)}),$$
where $i\ge 1.$ It is easy to see that the treasure sequences
$Q(\lambda)$ are monotonely decreasing, $\lim_{i \to +\infty} q_i=0.$
We call the first element or term $q_i<1$ in $Q(\lambda)$  a {\sl golden ball}, where $i\ge 1.$
Similarly,
we call the fixed index $i$ of the {\sl golden ball} of $Q(\lambda)$ as the {\sl golden point}, where $\lambda \gg 1$ is an arbitrary large, but fixed,  integer constant.

The following Lemma follows directly.

\begin{lemma}
There exists a golden ball in the treasure sequence.
\end{lemma}

The property of the treasure sequence will be used to further reduce the length of
the cover time $p(n)$ of the SUESs later.

\subsubsection{Construction of SUESs}

Same as in Section \ref{wm},
let $U_n$ be a sequence of length $n$ which is a UES for $d$-regular
graphs of size at most $bn^{\frac{1}{c}},$ for some constants
$b\ge 1,$ and $c > 1.$  And for every $k = 2^i$ and $n = 2^j$, where
$i < j,$ $U_k$ is a prefix of $U_n$.

We now define recursively a sequence  $S_n$ of strongly
universal exploration sequences. We start with
$S_1 = U_1.$ Assume that $U_n = u_1u_2\cdots u_n$ and that $n\ge 2.$
Define,
$$ S_n =u_1 S_{r_1} 0 S_{r_1}^{-1} 0 u_2 S_{r_2} 0 S_{r_2}^{-1} 0 u_3\cdots
u_i S_{r_i} 0 S_{r_i}^{-1} 0 u_{i+1}\cdots u_{n-1} S_{r_{n-1}} 0 S_{r_{n-1}}^{-1}0 u_n,$$

\noindent
where for every $1\le i < n,$ we set $r_i =\langle\langle i  \rangle\rangle$,
where $\langle\langle i\rangle\rangle = max\{2^j |2^{{\frac{\lambda+1}{\lambda}}j} \le i, j\in Z^{+}\}$. Note that as $n=2^j$, for some $j\ge 1,$ for every $k=2^p,$ where $p<j,$ the sequence $S_k$ is constructed as a prefix of $S_n$.
Moreover, we also assign $r_i=r_{n-i}$ for every $1\le i <n$. 

Note that $r_i\le \sqrt[{\lambda+1}]{i^{\lambda}}$.  Thus, if
$r_{\frac{n}{2}}=\sqrt[{\lambda+1}]{({\frac{n}{2}})^{\lambda}}$, the first half of $S_n$ is equal to
$S_{\frac{n}{2}}S_{\sqrt[{\lambda+1}]{({\frac{n}{2}})^{\lambda}}}0,$ and ends with a full copy
of $S_{\sqrt[{\lambda+1}]{({\frac{n}{2}})^{\lambda}}}$, followed by a $0$. Similarly, the second
half of $S_n$ starts with a full copy of $S_{\sqrt[{\lambda+1}]{({\frac{n}{2}})^{\lambda}}}^{-1}.$

We next bound the length of $S_n$.

Let $q^{[\lambda]}_g,g$
denote the {\sl golden ball} and {\sl golden point} of the treasure sequence $Q(\lambda)$, respectively.
Further more, if $g\ge 2$, we set  $C_{finite}=\sum_{j=1}^{g-1} \frac{2(|S_{2^j}|-|S_{2^{j-1}}|)}
{2^{\frac{(\lambda+1)j}{\lambda}}},$ and
$C_{max}=max\{y|y=\frac{S_{2^j}}{2^j},~for~1\le j\le g-1\}$, otherwise $C_{finite}=C_{max}=0$ (e.g. $g=1$).
Consequently, $q^{[\lambda]}_g,
g,C_{finite},$ and $C_{max}$ are constants due to the definition of the treasure sequence and the fact that only a
finite  number of terms are involved in the calculations, where $0<q^{[\lambda]}_g<1$, and $g\ge 1$.

We are now ready to prove the following Lemma.

\begin{lemma}
For every $n = 2^j,$ where $j\ge 1,$
$|S_n|<(\frac{5+C_{finite}}{1-q^{[\lambda]}_g}+C_{max})n.$
\end{lemma}

\begin{proof}
Let $s_n = |S_n|.$ For every $n\le 2^{g-1}$, we know it is true due to the definition of the constant
$C_{max}.$
The claim that $s_n <(\frac{5+C_{finite}}{1-q^{[\lambda]}_g}+C_{max})n$ for every
$n\ge 2^g$ then follows by using the induction.
It is not difficult to see that
\begin{eqnarray*}
|S_{2^i}| &=& |U_{2^i}|+ (|U_{2^i}|-1)\cdot 2(|S_1|+1)+
\sum_{j=1}^{\lfloor\frac{\lambda}{\lambda+1}\rfloor\cdot i-1} \lfloor\frac{|U_{2^i}|-1}{2^{\frac{(\lambda+1)j}{\lambda}}}\rfloor
\cdot 2(|S_{2^j}|-|S_{2^{j-1}}|) \\
s_{2^i} &\le& 2^i+ 4\cdot 2^i + 2^i\cdot
\sum_{j=1}^{\lfloor\frac{\lambda}{\lambda+1}\rfloor\cdot i-1} \frac{2(s_{2^j}-s_{2^{j-1}})}{2^{\frac{(\lambda+1)j}{\lambda}}}
 \\
&=& 2^i+ 4\cdot 2^i +
2^i\cdot
\sum_{j=1}^{g-1} \frac{2(s_{2^j}-s_{2^{j-1}})}{2^{\frac{(\lambda+1)j}{\lambda}}}+
2^i\cdot
\sum_{j=g}^{\lfloor\frac{\lambda}{\lambda+1}\rfloor\cdot i-1} \frac{2(s_{2^j}-s_{2^{j-1}})}{2^{\frac{(\lambda+1)j}{\lambda}}}
 \\
&<& 5\cdot 2^i +
2^i\cdot
\sum_{j=1}^{g-1} \frac{2(s_{2^j}-s_{2^{j-1}})}{2^{\frac{(\lambda+1)j}{\lambda}}}+
2^i\cdot
\sum_{j=g}^{+\infty} \frac{2(s_{2^j}-s_{2^{j-1}})}{2^{\frac{(\lambda+1)j}{\lambda}}}
 \\
&\le& 5\cdot 2^i + C_{finite}\cdot 2^i + 2^i\cdot
\sum_{j=g}^{+\infty}
\frac{2(s_{2^{j-1}}+2s_{
2^{\frac{(j-1)\lambda}{\lambda+1}}}+2)}{2^{\frac{(\lambda+1)j}{\lambda}}}
\\
&<& 5\cdot 2^i + C_{finite}\cdot 2^i + 2^i\cdot
(\frac{5+C_{finite}}{1-q^{[\lambda]}_g}+C_{max})\cdot
\sum_{j=g}^{+\infty}
\frac{2({2^{j-1}}+2\cdot{ 2^{\frac{(j-1)\lambda}{\lambda+1}}}+2)}{2^{\frac{(\lambda+1)j}{\lambda}}}
\\
&=& 5\cdot 2^i + C_{finite}\cdot 2^i + 2^i\cdot
(\frac{5+C_{finite}}{1-q^{[\lambda]}_g}+C_{max})\cdot q^{[\lambda]}_g\\
&<& 5\cdot 2^i + C_{finite}\cdot 2^i + 2^i\cdot C_{max}\cdot (1-q^{[\lambda]}_g)+
2^i\cdot
(\frac{5+C_{finite}}{1-q^{[\lambda]}_g}+C_{max})\cdot q^{[\lambda]}_g\\
&=& 2^i\cdot
(\frac{5+C_{finite}}{1-q^{[\lambda]}_g}+C_{max}).
\end{eqnarray*}
\end{proof}

Using the same arguments as in Lemma~\ref{3:2}, we can prove the following Lemma.

\begin{lemma}\label{p+1:p}

Let $k$ and $n\ge 2 k^{\frac{\lambda+1}{\lambda}}$ be powers of $2.$ Then,
every subsequence $T$ of $S_n$ or $S_n^{-1}$
of length $s_{2k^{\frac{\lambda+1}{\lambda}}}+1 = O(k^{\frac{\lambda+1}{\lambda}})$ contains, as a contiguous subsequence,
a full of $S_k$ or  $0 S^{-1}_k.$

\end{lemma}

\begin{proof}

We prove the claim by induction on $n.$ If
$n = 2k^{\frac{\lambda+1}{\lambda}}$ then the claim is vacuously satisfied as $S_n$ contains a full $S_k$.

Assume, therefore, that the claim holds for every
$m = 2^{j'}$
that satisfies $2k^{\frac{\lambda+1}{\lambda}}\le m < n = 2^j$. We show
that it also holds for $n.$ Let $T$ be a subsequence of $S_n$
of length $s_{2k^{\frac{\lambda+1}{\lambda}}} +1.$ Essentially the same argument works
if $T$ is a subsequence of such length of $S^{-1}_n$.

Same as in Lemma~\ref{3:2}, we study the following cases:

\noindent
{\bf Case 1:} T is completely contained in a subsequence $S_m$
or $S^{-1}_m$ of $S_n,$ for some $m < n.$

The claim then follows immediately from the induction
hypothesis.

\noindent
{\bf Case 2:} $T$ is completely contained in a subsequence
$S_m0S^{-1}_m$ of $S_n,$ for some $m < n.$

In this case, $T = T'0T''$, where $T'$ is a suffix of $S_m$
and $T''$ is a prefix of $S^{-1}_m.$ Either
$|T'|\ge {\frac{1}{2}}s_{2k^{\frac{\lambda+1}{\lambda}}}$ or
$|T''|\ge {\frac{1}{2}}s_{2k^{\frac{\lambda+1}{\lambda}}}$.
Assume that
$|T''|\ge {\frac{1}{2}}s_{2k^{\frac{\lambda+1}{\lambda}}}$.
Another case is analogous. As $T''$ is a prefix of $S^{-1}_m$,
and  $|T''|\ge {\frac{1}{2}}s_{2k^{\frac{\lambda+1}{\lambda}}}$, it follows that
$m\ge 2k^{\frac{\lambda+1}{\lambda}}$. Now, $S^{-1}_k$
is almost a prefix of $S^{-1}_m$, in the sense that they differ
only in symbols that originate directly from $S_m.$
In particular, a prefix of $S^{-1}_m$ of length
${\frac{1}{2}}s_{2k^{\frac{\lambda+1}{\lambda}}}$, half the
length of $S_{2k^{\frac{\lambda+1}{\lambda}}}$, ends with a full copy of $S_k,$
followed by $0.$

\noindent
{\bf Case 3:} $T$ contains a symbol $u_l$ of $S_n$ that originates
from $U_n.$

In this case, $T = T'u_lT''.$ Again, we have either
$|T'|\ge {\frac{1}{2}}s_{2k^{\frac{\lambda+1}{\lambda}}}$ or
$|T''|\ge {\frac{1}{2}}s_{2k^{\frac{\lambda+1}{\lambda}}}$. Assume again
that $|T''|\ge {\frac{1}{2}}s_{2k^{\frac{\lambda+1}{\lambda}}}$.
Another case is analogous. Let
$$S_{n,l}= u_l S_{r_l}0S^{-1}_{r_l}0 u_{l+1}\cdots u_{n-1} S_{r_{n-1}}
0S^{-1}_{r_{n-1}}0 u_n$$
\noindent
be the suffix of $S_n$ that starts with the symbol $u_l$ that
originates from the $l$-th symbol of $U_n.$ We claim that
the prefix of $S_{n,l}$ of length ${\frac{1}{2}}s_{2k^{\frac{\lambda+1}{\lambda}}}$
contains a copy of $S_k.$
Let $l'=\lceil\frac{l}{k^{\frac{\lambda+1}{\lambda}}}\rceil k^{\frac{\lambda+1}{\lambda}}$ be
the first index after $l$ which is
divisible by $k^{\frac{\lambda+1}{\lambda}}.$ Clearly $r_{l'}\ge k$ and hence $S_k$ is a
prefix of $S_{r_{l'}}$. Thus,
$S' = u_l S_{r_l}0S^{-1}_{r_l}0\cdots u_{l'}S_k$ is a
prefix of $S_{m,l}$ which ends with a complete $S_k.$ As
for every $l\le i <l'$ we have $r_i = r_{i~mod~k^{\frac{\lambda+1}{\lambda}}},$ we have
that $S'$ is contained in the first half of $S_{2k^{\frac{\lambda+1}{\lambda}}},$ and hence
$|S'|\le {\frac{1}{2}}s_{2k^{\frac{\lambda+1}{\lambda}}}$ as expected.

\end{proof}

Furthermore, by the same arguments as in Theorem~\ref{wt}, we have:

\begin{theorem}

If for every $n\ge 1$ of the power of $2$ there is an UES of
length $O(n^c)$ for $d$-regular graphs of size at most $n$,
then there exists an infinite SUES for $d$-regular graphs with
cover time $p(n)=O(n^{c(1+\frac{1}{\lambda})})$,
where $c$ is the constant induced from the construction of an universal exploration
sequence in \cite{R05,TZ07}, and $\lambda \gg 1$ is an arbitrary large, but fixed,  integer constant. 
Moreover, the SUESs can be constructed deterministically in polynomial time.

\end{theorem}

\begin{proof}

Let us ignore all the recursive components of $S_j$ from $S_n$ such as $j < n$,
and their inverses, which because that
$0S^{-1}_j0$ reverses the actions of $S_j.$ The left parts are
$U_n=u_1,u_2,u_3,\cdots, u_n.$ Moreover, note that $U_n$ is a UES for
for $d$-regular graphs of size at
least $bn^{\frac{1}{c}},$ for some constants $b\ge 1$ and $c>1$
due to Theorem~\ref{R05}. Theorem~\ref{R05} also show that such a
UES could be constructed, deterministically, in polynomial time.
According to Lemma~\ref{p+1:p}, we know that
every subsequence $T$ of the SUES we constructed of length
$s_{2n^{\frac{\lambda+1}{\lambda}}}+1 = O(n^{\frac{\lambda+1}{\lambda}})$ contains,
as a contiguous subsequence,
a full copy of $S_n$. Consequently,
there is an infinite SUES for $d$-regular graphs with
cover time $p(n)=O(n^{{\frac{\lambda+1}{\lambda}}c})$, where $c$
is the fixed constant from the construction of an universal exploration
sequence in \cite{R05,TZ07}.
Furthermore, the SUESs can be constructed deterministically in polynomial time.

\end{proof}

Finally, by employing the standard double techniques in $d$-regular graphs of size at most $n$, we get the desired result.

\begin{theorem}

If for every $n\ge 1$ there is an UES of
length $O(n^c)$ for $d$-regular graphs of size at most $n$,
then there exists an infinite SUES for $d$-regular graphs with
cover time $p(n)=O(n^{c(1+\frac{1}{\lambda})})$,
where $c$ is the fixed constant induced from the construction of an universal exploration
sequence in \cite{R05,TZ07}, and $\lambda \gg 1$ is an arbitrary large, but fixed,  integer constant.
Moreover, the SUESs can be constructed deterministically in polynomial time.

\end{theorem}

\begin{remark1}
It is easy to
extend the solutions given for the $d$-regular graphs to
general graphs by using the ideas from \cite{DFKP06,TZ07}.
\end{remark1}

\begin{remark1}
The proposed explicit {\sl SUESs} could be also used to improve the running time of the explicit solution suggested for the rendezvous problem with backtracking in \cite{TZ07}.
\end{remark1}

\section{Conclusion and open problems}

We proposed an improved explicit deterministic solution for the treasure hunt problem with backtracking. More precisely,
we derived an $O(n^{c(1+\frac{1}{\lambda})})$-time algorithm for the treasure hunt,  which significantly
improves the currently best known result with running time $O(n^{2c})$ in \cite{TZ07},
where $c$ is the constant induced from the construction of an universal exploration
sequence in \cite{R05,TZ07}, $\lambda \gg 1$ is an arbitrary large, but fixed, integer constant. In this work, 
we also proposed a much better explicit construction 
for strongly universal exploration sequences compared to the one in \cite{TZ07}. The proposed explicit {\sl SUESs} could be also  used to further improve the time complexity of the explicit solution addressed for the rendezvous problem with backtracking in \cite{TZ07}.

The existence of strongly universal exploration sequences without backtracking
is left as an intriguing open problem. \\

\end{document}